\documentclass{birkjour}
 \newtheorem{thm}{Theorem}[section]
 \newtheorem{cor}[thm]{Corollary}
 \newtheorem{lem}[thm]{Lemma}
 
 \theoremstyle{definition}
 
 \theoremstyle{remark}

 \numberwithin{equation}{section}

\def\Spin{{\rm Spin}}
\def\Pin{{\rm Pin}}
\def\Cen{{\rm Cen}}
\def\Mat{{\rm Mat}}
\def\R{{\mathbb R}}
\def\C{{\mathbb C}}
\def\OO{{\rm O}}
\def\SO{{\rm SO}}

\def\sign{{\rm sign}}
\def\End{{\rm End}}
\def\T{{\rm T}}
\def\cl{{C}\!\ell}

\begin{document}

%-------------------------------------------------------------------------
% editorial commands: to be inserted by the editorial office
%
%\firstpage{1} \volume{228} \Copyrightyear{2004} \DOI{003-0001}
%\seriesextraline{This is the Concrete Title of this Book\br H.E. R and S.T.C. W, Eds.}
%\firstpage{1}

\title[Calculation of elements of spin groups using method of averaging]
 {Calculation of elements of spin groups\\ using method of averaging in Clifford's\\ geometric algebra}

%----------Author 1
\author[Dmitry Shirokov]{Dmitry Shirokov}

\address{%
National Research University Higher School of Economics\\
Myasnitskaya str. 20\\
101000 Moscow\\
Russia}
\address{
Institute for Information Transmission Problems of Russian Academy of Sciences\\
Bolshoy Karetny per. 19\\
127051 Moscow\\
Russia}
\email{dm.shirokov@gmail.com}


%\subjclass{Primary 15A66; Secondary 11E88, 15B10, 20B05}
\subjclass{15A66, 11E88, 15B10, 20B05}

\keywords{spin group, Clifford algebra, geometric algebra, rotor, method of averaging, orthogonal group, two-sheeted cover}

\date{December 31, 2018}

\begin{abstract}
We present a method of computing elements of spin groups in the case of arbitrary dimension. This method generalizes Hestenes method for the case of dimension 4. We use the method of averaging in Clifford's geometric algebra previously proposed by the author. We present explicit formulas for elements of spin group that correspond to the elements of orthogonal groups as two-sheeted covering. These formulas allow us to compute rotors, which connect two different frames related by a rotation in geometric algebra of arbitrary dimension.
\end{abstract}

%%% ----------------------------------------------------------------------
\maketitle
%%% ----------------------------------------------------------------------
%\tableofcontents
\section{Introduction}
\label{sec1}

In Clifford's geometric algebra, it is convenient to describe rotations using elements of spin groups. Spin groups of arbitrary dimension are naturally realized in this algebra. Nowadays Clifford's geometric algebra is widely used in physics, computer science, engineering, and other sciences.

Professor D. Hestenes presented the method of computing elements of spin group $\Spin_+(1,3)$ in the case of dimension $n=4$ in \cite{Hestenes} (pp.~52--53). This method is mentioned in other papers and books (see, for example, \cite{Lounesto}, p.~130). We generalize this method to the case of arbitrary dimension $n=p+q$ and all spin groups $\Spin_+(p,q)$, $\Spin(p,q)$, $\Pin_+(p,q)$, $\Pin_-(p,q)$, $\Pin(p,q)$. We use the method of averaging in Clifford's geometric algebra previously proposed by the author \cite{averaging}, \cite{averaging2} to do this.

There are other methods of calculation of elements of spin groups using exponentials and exterior exponentials of bivectors (see \cite{Hestenes}, \cite{Hestenes2}, \cite{Lasenby}, \cite{Lounesto}, \cite{Marchuk}). But all these methods work only in the cases of fixed dimensions, especially $n=3$ or $n=4$. In this paper, we present explicit formulas for the elements of spin groups, which work in the case of arbitrary $n$.

Note that some years ago we presented another method of computing elements of spin groups using generalized Pauli's theorem. These results were presented at the conference AGACSE 2012 (La Rochelle, France, July 2012) and published in the Conference Proceedings in AACA \cite{spin}. In \cite{spin}, we presented an algorithm (we had no explicit formulas) for computing elements of spin groups. In the new method presented in the current paper, using the method of averaging, we obtain explicit formulas for the elements of spin groups.

The paper is organized as follows. In Section \ref{sec2}, we discuss a formalism of Clifford's geometric algebra and use it for the consideration of pseudo-orthogonal groups. We prove some auxiliary lemmas. In Section \ref{sec3}, we present a complete picture of five orthogonal groups and five corresponding spin groups in the case of arbitrary dimension. In Section \ref{sec4}, we discuss Hestenes method of computing elements of the group $\Spin_+(1,3)$. In Section \ref{sec5}, we generalize this method for the case of arbitrary dimension using the method of averaging in Clifford's geometric algebra. In Section \ref{sec6}, we present an explicit formula for computing rotors, which connect two different frames related by a rotation in geometric algebra of arbitrary dimension.

\section{Clifford algebras and pseudo-orthogonal group}
\label{sec2}

Let us consider the real Clifford algebra $\cl_{p,q}$, $p+q=n$, with the identity element $e$ and the generators $e_a$, $a=1, \ldots, n$, satisfying
\begin{eqnarray}
e_a e_b+e_b e_a=2\eta_{ab}e,\label{gen}
\end{eqnarray}
where $\eta=||\eta_{ab}||$ is the diagonal matrix with its first $p$ entries equal to $1$ and the last $q$ entries equal to $-1$ on the diagonal.

We use notation with ordered multi-indices $A$ for the basis elements of the Clifford algebra $\cl_{p,q}$:
$$e_A=e_{a_1 \ldots a_k},\qquad 1\leq a_1 < \cdots < a_k \leq n.$$
We denote the length of multi-index $A$ by $|A|$. In the case of the identity element $e$, we have empty multi-index $\o$ of length $0$. We call the subspace of $\cl_{p,q}$ of Clifford algebra elements, which are linear combinations of basis elements with multi-indices of length $|A|=k$, the subspace of grade $k$ and denote it by $\cl^k_{p,q}$. We denote the projection operator onto subspace of the grade $k$ by $\pi_k$. We denote inverses of generators by $e^a:=\eta^{ab}e_b=(e_a)^{-1}$, $a=1, \ldots, n$, and inverses of basis elements by $e^A=(e_A)^{-1}$. The frame $e^a$, $a=1, \ldots, n$ is often called reciprocal frame for the frame $e_a$, $a=1, \ldots, n$.

Even and odd subspaces we denote by $\cl^{(0)}_{p,q}$ and $\cl^{(1)}_{p,q}$. We have
$$\cl_{p,q}=\bigoplus_{k=0}^n \cl^k_{p,q},\qquad \cl^{(0)}_{p,q}=\!\!\!\bigoplus_{k=0\!\!\mod 2}\cl^k_{p,q},\qquad \cl^{(1)}_{p,q}=\!\!\!\bigoplus_{k=1\!\!\mod 2}\cl^k_{p,q}.$$

Let us consider the pseudo-orthogonal group $\OO(p,q)$, $p+q=n$:
\begin{eqnarray}
\OO(p,q):=\{P\in\Mat(n,\R): P^\T \eta P=\eta\}.\label{oo}
\end{eqnarray}
Denote by
$$p^A_B=p^{a_1 \ldots a_k}_{b_1 \ldots b_k},\qquad a_1 < \cdots < a_k,\qquad b_1 < \cdots < b_k,$$
the minors of the matrix $P=||p^a_b||$. They are determinants of the submatrices formed by rows $a_1$, \ldots, $a_k$ and columns $b_1$, \ldots, $b_k$. In the particular case of multi-indices of length $1$ ($A=a$, $B=b$), the corresponding minor is just an element $p^a_b$ of the matrix $P$. In the case of empty multi-indices $A$ and $B$, the corresponding minor equals $1$ by definition. We use Einstein summation convention for ordered multi-indices too.

We have the following simple fact.
\begin{lem}\label{lemma1}
The set $\beta_a:=p_a^b e_b\in\cl_{p,q}$ satisfies the following conditions
$$\beta_a \beta_b +\beta_b \beta_a=2\eta_{ab}e$$
if and only if
$$P=||p_a^b||\in\OO(p,q)=\{P\in\Mat(n, \R): P^\T \eta P=\eta\}.$$
\end{lem}
\begin{proof}
We have
$$\beta_a \beta_b+\beta_b \beta_a=p_a^c p_b^d (e_c e_d+e_d e_c)=2 p_a^c p_b^d \eta_{cd}.$$
This means that $\beta_a \beta_b +\beta_b \beta_a=2\eta_{ab}e$ if and only if $p_a^c p_b^d \eta_{cd}=\eta_{ab}$. From the last condition, we obtain the definition of $\OO(p,q)$ (\ref{oo}).
\end{proof}

For
\begin{eqnarray}
\beta_a:=p_a^b e_b,\qquad P=||p_a^b||\in\OO(p,q),\label{beta}
\end{eqnarray}
let us consider the elements
\begin{eqnarray}
\beta_A=\beta_{a_1 \ldots a_k}:=\beta_{a_1}\cdots \beta_{a_k},\qquad 1\leq a_1 < \cdots < a_k \leq n.\label{newbas}
\end{eqnarray}

\begin{lem}\label{lemma2}
For (\ref{beta}), we have
\begin{eqnarray}
\beta_{a_1 \ldots a_k}=p_{a_1 \ldots a_k}^{b_1 \ldots b_k}e_{b_1 \ldots b_k},\label{min}
\end{eqnarray}
where $p_{a_1 \ldots a_k}^{b_1 \ldots b_k}$ are minors of the matrix $P=||p_a^b||\in\OO(p,q)$ and we have a sum over all ordered multi-indices $b_1 \ldots b_k$ of length $k$ in (\ref{min}).
\end{lem}
Using our notations, we can rewrite (\ref{min}) in the following way
\begin{eqnarray}
\beta_A=p_A^B e_B,\label{min2}
\end{eqnarray}
where we have a sum over all ordered multi-indices $B$ of the same length as the length of the multi-index $A$.
\begin{proof}
For $k=1$, we have $\beta_a=p_a^b e_b$ by the definition. For $k=2$, we have
\begin{eqnarray}
\beta_{a_1 a_2}&=&\beta_{a_1} \beta_{a_2}=(p_{a_1}^1 e_1+\cdots+p_{a_1}^n e_n)(p_{a_2}^1 e_1+\cdots+p_{a_2}^n e_n)\nonumber\\
&=&(p_{a_1}^{1} p_{a_2}^1 \eta_{11}+\cdots+p_{a_1}^n p_{a_2}^n \eta_{nn})e\nonumber\\
&+&(p_{a_1}^1 p_{a_2}^2-p_{a_2}^1 p_{a_1}^2)e_{12}+\cdots+(p_{a_1}^{n-1} p_{a_2}^n-p_{a_2}^{n-1} p_{a_1}^n)e_{n-1 \, n}\nonumber\\
&=&p_{a_1 a_2}^{12}e_{12}+\cdots+ p_{a_1 a_2}^{n-1 \, n}e_{n-1\, n}=p_{a_1 a_2}^{b_1 b_2}e_{b_1 b_2},\nonumber
\end{eqnarray}
where we use
\begin{eqnarray}
p_{a_1}^{b}p_{a_2}^{b} \eta_{b b}=0,\qquad a_1<a_2,\qquad P=||p_a^b||\in\OO(p,q).\label{po3}
\end{eqnarray}

In the general case, the proof is by induction on $k$. Suppose that we have (\ref{min}) for $k=m-1$. Let us prove it for $k=m$. We have
\begin{eqnarray}
\beta_{a_1 \ldots a_m}=\beta_{a_1 \ldots a_{m-1}}\beta_{a_m}=(p_{a_1 \ldots a_{m-1}}^{c_1 \ldots c_{m-1}}e_{c_1 \ldots c_{m-1}})(p_{a_m}^{c}e_{c}).\label{po}
\end{eqnarray}
Multiplying two sums, we obtain the element of grade $m-2$ (in the case $c \in \{c_1, \ldots, c_{m-1}\}$) and the element of grade $m$ (in the case $c \notin \{c_1, \ldots, c_{m-1}\}$).

The corresponding element of grade $m-2$ equals zero. Using (\ref{po3}) and Laplace expansion, we can prove it again by induction. We omit detailed proof because of its cumbersomeness.

The corresponding element of grade $m$ equals $p_{a_1 \ldots a_m}^{b_1 \ldots b_m}e_{b_1 \ldots b_m}$ because of the Laplace expansion along one column of the corresponding minor:
$$
p_{a_1 \ldots a_m}^{b_1 \ldots b_m}=\sum_{j=1}^{m}(-1)^{m+j} p_{a_1 \ldots a_{m-1}}^{b_1 \ldots \check{b_j} \ldots b_{m}}p_{a_m}^{b_j},
$$
where $b_1 \ldots \check{b_j} \ldots b_{m}$ is the ordered multi-index of length $m-1$, which is obtained from $b_1 \ldots b_m$ by discarding $b_j$.
\end{proof}

Note that as particular case of (\ref{min}) we get
$$\beta_{1\ldots n}=\det (P) e_{1\ldots n},\qquad \det P=\pm 1.$$
The conditions $\beta_{1\ldots n}=\pm e_{1\ldots n}$ mean that $2^n$ elements $\beta_A$ (\ref{newbas}) are linearity independent and constitute a new basis of $\cl_{p,q}$ (see pp. 289--290 in \cite{Snygg} or pp. 127--128 in \cite{Port}).

\section{Complete picture of orthogonal and spin groups}
\label{sec3}

For the convenience of the reader, we present a complete picture of five orthogonal groups and the corresponding five spin groups in the case of arbitrary dimension. For more details, see \cite{lect} and \cite{BT}. Often, only the group $\Spin_+(p,q)$ is considered, but sometimes other spin groups are also required for different applications.

\begin{lem}
For an arbitrary matrix $P\in\OO(p,q)$, we have
\begin{eqnarray}
\det P=\pm 1,\quad |p^{1\ldots p}_{1\ldots p}|\geq 1,\quad |p^{p+1 \ldots ,n}_{p+1 \ldots n}|\geq 1,\quad p^{1\ldots p}_{1\ldots p}=\frac{p^{p+1 \ldots n}_{p+1\ldots n}}{\det P},\label{4t}
\end{eqnarray}
where $p^{1\ldots p}_{1\ldots p}$ and $p^{p+1 \ldots ,n}_{p+1 \ldots n}$ are the corresponding minors of the matrix $P$.
\end{lem}
\begin{proof} The first statement is trivial. For the matrix
$$P=\left(
      \begin{array}{cc}
        A_{p\times p} & B_{p \times q} \\
        C_{q\times p} & D_{q \times q} \\
      \end{array}
    \right)\in\OO(p,q)$$
with the blocks $A, B, C, D$ of corresponding sizes, we have
$$P^\T \eta P= \eta,\quad B^\T B-D^\T D=-{\bf 1},\quad |\det(D)|=|p^{p+1 \ldots ,n}_{p+1 \ldots n}|\geq 1;$$
$$P \eta P^\T =\eta,\quad A A^\T- B B^\T={\bf 1},\quad |\det(A)|=|p^{1\ldots p}_{1\ldots p}|\geq 1.$$
From $\eta P^\T \eta= P^{-1}$ and the well-known formula for the minor of the inverse of a matrix (see, for example, \cite{Gant}, pp. 25 - 27), we obtain
$$p^{1\ldots p}_{1\ldots p}=(p^{-1})^{1\ldots p}_{1\ldots p}=\frac{p^{p+1 \ldots n}_{p+1\ldots n}}{\det P},$$
where $(p^{-1})^{1\ldots p}_{1\ldots p}$ is the corresponding minor of the matrix $P^{-1}$.
\end{proof}

The group $\OO(p,q)$ has the following subgroups:
\begin{eqnarray}
\OO_+(p,q)&:=&\{P\in\OO(p,q): p^{1\ldots p}_{1\ldots p}\geq 1\},\nonumber\\
\OO_-(p,q)&:=&\{P\in\OO(p,q): p^{p+1\ldots n}_{p+1\ldots n}\geq 1\},\nonumber\\
\SO(p,q)&:=&\{P\in\OO(p,q): \det P=1\},\nonumber\\
\SO_+(p,q)&:=&\{P\in\SO(p,q): p^{1\ldots p}_{1\ldots p}\geq 1\}=\{P\in\SO(p,q): p^{p+1\ldots n}_{p+1\ldots n}\geq 1\}.\nonumber
\end{eqnarray}
For example, in the particular case of Minkowski space, we have Lorentz group $\OO(1,3)$, special (or proper) Lorentz group $\SO(1,3)$, orthochronous (or time preserving) Lorentz group $\OO_+(1,3)$, orthochorous (or parity preserving) Lorentz group $\OO_-(1,3)$, special orthochronous Lorentz group $\SO_+(1,3)$.

In Euclidean cases ($p=0$ or $q=0$), we have only two orthogonal groups instead of five groups:
\begin{eqnarray}
\OO(n)&:=&\OO(n,0)=\OO_-(n,0)\cong \OO(0,n)=\OO_+(0,n),\nonumber\\
\SO(n)&:=&\SO(n,0)=\SO_+(n,0)=\OO_+(n,0)\nonumber\\
&\cong&\SO(0,n)=\SO_+(0,n)=\OO_-(0,n).\nonumber
\end{eqnarray}

We denote grade involution (main involution) in $\cl_{p,q}$ by
$$\widehat{U}:=U|_{e_a\to -e_a},\qquad U\in\cl_{p,q}$$
and reversion (anti-involution) by
$$\widetilde{U}:=U|_{e_{a_1\ldots a_k}\to e_{a_k}\ldots e_{a_1}},\qquad U\in\cl_{p,q}.$$

Denote by $M^\times$ the subset of invertible elements of any set $M$. Let us consider the Lipschitz group
\begin{eqnarray}
\Gamma^\pm_{p,q}&:=&\{S\in\cl^{(0)\times}_{p,q}\cup\cl^{(1)\times}_{p,q}: S \cl^1_{p,q}S^{-1}\subset \cl^1_{p,q}\}\nonumber\\
&=&\{v_1 \cdots v_k: v_1, \ldots, v_k \in \cl^{1\times}_{p,q}\}\nonumber
\end{eqnarray}
and its subgroup
\begin{eqnarray}
\Gamma^+_{p,q}&:=&\{S\in\cl^{(0)\times}_{p,q}: S \cl^1_{p,q} S^{-1}\subset\cl^1_{p,q}\}\nonumber\\
&=&\{v_1 \cdots v_{2k}: v_1, \ldots, v_{2k} \in \cl^{1\times}_{p,q}\}\subset \Gamma^\pm_{p,q}.\nonumber
\end{eqnarray}
The following groups are called spin groups:
\begin{eqnarray}
\Pin(p,q)&:=&\{ S\in\Gamma^\pm_{p,q}: \widetilde{S} S=\pm e\}=\{ S\in\Gamma^\pm_{p,q}: \widehat{\widetilde{S}} S=\pm e\},\nonumber\\
\Pin_+(p,q)&:=&\{S\in\Gamma^\pm_{p,q}: \widehat{\widetilde{S}} S=+e\},\nonumber\\
\Pin_-(p,q)&:=&\{S\in\Gamma^\pm_{p,q}: \widetilde{S} S=+e\},\label{spingr}\\
\Spin(p,q)&:=&\{S\in\Gamma^+_{p,q}: \widetilde{S} S= \pm e\}= \{S\in\Gamma^+_{p,q}: \widehat{\widetilde{S}} S= \pm e\},\nonumber\\
\Spin_+(p,q)&:=&\{S\in\Gamma^+_{p,q}: \widetilde{S}S=+e\}=\{S\in\Gamma^+_{p,q}: \widehat{\widetilde{S}} S=+e\}.\nonumber
\end{eqnarray}
Let us consider the twisted adjoint representation
$$\phi: \cl^\times_{p,q}\to \End\cl_{p,q},\qquad S \to \phi_S,\qquad \phi_S U=\widehat{S}US^{-1},\qquad U\in\cl_{p,q}.$$
The following homomorphisms are surjective with the kernel $\{\pm 1\}$:
\begin{eqnarray}
&&\phi: \Pin(p,q) \to \OO(p,q),\nonumber\\
&&\phi: \Spin(p,q) \to \SO(p,q),\nonumber\\
&&\phi: \Spin_+(p,q) \to \SO_+(p,q),\nonumber\\
&&\phi: \Pin_+(p,q) \to \OO_+(p,q),\nonumber\\
&&\phi: \Pin_-(p,q) \to \OO_-(p,q).\nonumber
\end{eqnarray}
It means that, for all $P=||p^a_b||\in\OO(p,q)$, there exist $\pm S\in\Pin(p,q)$ such that
\begin{eqnarray}
\widehat{S} e_a S^{-1}=p_a^b e_b\label{sv}
\end{eqnarray}
and for the other groups similarly. The spin groups (\ref{spingr}) are two-sheeted coverings of the corresponding orthogonal groups.

Our goal is to find out elements $\pm S\in\Pin(p,q)$ for each $P\in\OO(p,q)$ in the case of arbitrary $p$ and $q$ using the relation (\ref{sv}).

\section{Hestenes method}
\label{sec4}

Let us consider the method proposed by D. Hestenes \cite{Hestenes} for the case of dimension $n=4$, $\cl_{1,3}$.

For each element $P=||p^a_b||\in\SO_+(1,3)$, there exist two elements $\pm S\in\Spin_+(1,3)$ such that
\begin{eqnarray}
Se_a S^{-1}=p_a^b e_b,\qquad S^{-1}=\tilde{S}.\label{sp}
\end{eqnarray}
The elements $\pm S$ can be found in the following way
\begin{eqnarray}
S=\pm\frac{L}{\sqrt{\widetilde{L}L}},\qquad L:=p^b_a e_b e^a.\label{Hest}
\end{eqnarray}

Let us discuss the plan of the proof of the formula (\ref{Hest}). Multiplying both sides of the first equation (\ref{sp}) on the right by $e^a=(e_a)^{-1}$, we obtain
\begin{eqnarray}
S e_a S^{-1} e^a=p_a^b e_b e^a=:L.\label{h1}
\end{eqnarray}
We have the following well-known formula (see \cite{Lounesto}, \cite{unitary})
\begin{eqnarray}
e_a \pi_k(S) e^a=(-1)^k (n-2k) \pi_k(S),\qquad S\in\cl_{p,q},\qquad k=0, \ldots, n.\label{h2}
\end{eqnarray}
We have $S\in\Spin_+(1,3)$, so $S=\pi_0(S)+\pi_2(S)+\pi_4(S)$. Using (\ref{h2}), we get from (\ref{h1})
\begin{eqnarray}
4S(\pi_0(S^{-1})-\pi_4(S^{-1}))=L.\label{h3}
\end{eqnarray}
Let us take reversion of both sides of (\ref{h3}). We get
\begin{eqnarray}
4(\pi_0(S^{-1})-\pi_4(S^{-1}))\widetilde{S}=\widetilde{L}.\label{h4}
\end{eqnarray}
Multiplying both sides of (\ref{h4}) by both sides of (\ref{h3}) and using $\widetilde{S}S=e$, we obtain
\begin{eqnarray}
(4(\pi_0(S^{-1})-\pi_4(S^{-1})))^2=\widetilde{L}L.\label{h5}
\end{eqnarray}
Both sides of this equation belong to $\cl^0_{1,3}\oplus\cl^4_{1,3}\cong\C$. Taking square root of both sides of (\ref{h5}), we get
\begin{eqnarray}
4(\pi_0(S^{-1})-\pi_4(S^{-1}))=\pm\sqrt{\widetilde{L}L}.\label{h6}
\end{eqnarray}
Substituting (\ref{h6}) into (\ref{h3}), we obtain (\ref{Hest}).

Note that this method works only in the case of dimension $n=4$ for the matrices $P=||p^a_b||\in\SO_+(1,3)$ with additional condition
\begin{eqnarray}
L \neq 0.\label{condH}
\end{eqnarray}
The condition (\ref{condH}) is equivalent to the condition
$$\pi_0(S)\neq 0 \qquad \mbox{or} \qquad \pi_4(S)\neq 0$$
for the corresponding element $S\in\Spin_+(1,3)$.

In the next section, we will generalize this method to the case of arbitrary $n=p+q$. We will use some other operators instead of (\ref{h2}) to do this.

\section{Generalization of Hestenes method}
\label{sec5}

We have the following new theorems.

\begin{thm}\label{th1} Let us consider the real Clifford algebra $\cl_{p,q}$ with even $n=p+q$. Let $P\in\SO(p,q)$ be an orthogonal matrix such that
\begin{eqnarray}
M:=\sum_{A, B}p^B_A e_B e^A\neq 0.\label{cond}
\end{eqnarray}
Then we can find the elements $\pm S\in\Spin(p,q)$ that correspond to $P=||p^b_a||\in\SO(p,q)$ as two-sheeted covering $S e_a S^{-1}=p_a^be_b$ in the following way:
\begin{eqnarray}
S=\pm\frac{M}{\sqrt{\alpha\,\widetilde{M}M}},\label{1}
\end{eqnarray}
where
$$\widetilde{M}M \in \Cen(\cl_{p,q})=\cl_{p,q}^0\cong \R$$
and the sign
$$\alpha:=\sign (p^{1\ldots p}_{1\ldots p})e=\sign(p^{p+1 \ldots n}_{p+1 \ldots n})e=\widetilde{S}{S}=\pm e$$ depends on the component of the orthogonal group $\SO(p,q)$ (or the corresponding component of the spin group $\Spin(p,q)$).
\end{thm}

\begin{thm}\label{th2} Let us consider the real Clifford algebra $\cl_{p,q}$ with odd $n=p+q$. Let $P\in\OO(p,q)$ be an orthogonal matrix such that
\begin{eqnarray}
M:=\sum_{A, B}(\det P)^{|A|}p^B_A e_B e^A\neq 0.\label{cond2}
\end{eqnarray}
Then we can find the elements $\pm S\in\Pin(p,q)$ that correspond to $P=||p^b_a||\in\OO(p,q)$ as two-sheeted covering $\widehat{S} e_a S^{-1}=p_a^be_b$ in the following way:
\begin{eqnarray}
S=\pm\frac{M}{\sqrt{\alpha\,\widetilde{M}M}},\label{2}
\end{eqnarray}
where
$$\widetilde{M}M\in\cl^0_{p,q}\subset\Cen(\cl_{p,q})\cong\left\lbrace
\begin{array}{ll}
\R\oplus\R, & \mbox{if $p-q=1\mod 4$;}\\
\C, & \mbox{if $p-q=3 \mod 4$}
\end{array}
\right.$$
and the sign
\begin{eqnarray}
\alpha:=\left\lbrace
\begin{array}{ll}
\sign(p^{p+1 \ldots n}_{p+1 \ldots n})e=\widetilde{S}{S}=\pm e, & \mbox{if $n=1\mod 4$;}\\
\sign(p^{1\ldots p}_{1\ldots p})e=\widehat{\widetilde{S}}S=\pm e, & \mbox{if $n=3 \mod 4$}
\end{array}
\right.\label{cas}
\end{eqnarray}
depends on the component of the orthogonal group $\OO(p,q)$ (or the corresponding component of the group $\Pin(p,q)$).
\end{thm}

The conditions (\ref{cond}) and (\ref{cond2}) are equivalent to the condition $\pi_{\Cen}(S)\neq 0$ for the corresponding element $S\in\Pin(p,q)$, where $\pi_{\Cen}$ is the projection onto the center of the Clifford algebra
$$\Cen(\cl_{p,q})=\{U\in\cl_{p,q}: UV=VU \,\mbox{for all}\, V\in\cl_{p,q}\}.$$

\begin{proof}
In the proof of Theorems \ref{th1} and \ref{th2}, we use Reynolds operators of Salingaros vee group (see \cite{averaging})
\begin{eqnarray}
F(U):=\frac{1}{2^n}e_A U e^A=\pi_{\Cen}(U),\qquad U\in\cl_{p,q}.\label{contr}
\end{eqnarray}

We have the following relation (\ref{sv}) between orthogonal matrix $P=||p^b_a||\in\OO(p,q)$ and the corresponding two elements of spin group $\pm S\in\Pin(p,q)$. We can rewrite this relation in the following way
\begin{eqnarray}
S e_a S^{-1}=(\det P) p_a^b e_b\label{t1}
\end{eqnarray}
because of the relation between parity of the element of spin group and the determinant of the corresponding orthogonal matrix (see Section \ref{sec3}). Multiplying (\ref{t1}) by itself several times and using Lemma \ref{lemma2}, we get
\begin{eqnarray}
S e_A S^{-1}=(\det P)^{|A|}p_A^B e_B.\label{t2}
\end{eqnarray}
Multiplying both sides of (\ref{t2}) on the right by $e^A=(e_A)^{-1}$, we obtain
\begin{eqnarray}
S e_A S^{-1} e^A=(\det P)^{|A|}p_A^B e_B e^A.\label{t3}
\end{eqnarray}
Here we have a sum over multi-indices $A$, $B$ of the same length $|A|=|B|$.
We denote the right side of (\ref{t3}) by
\begin{eqnarray}
M:=\sum_{A, B} (\det P)^{|A|}p_A^B e_B e^A\in\cl^{(0)}_{p,q}\label{MM}
\end{eqnarray}
and get
\begin{eqnarray}
2^n S\, \pi_{\Cen}(S^{-1})=M.\label{t4}
\end{eqnarray}
Using $M\in\cl^{(0)}_{p,q}$ and $S\in\cl^{(0)}_{p,q}\cup\cl^{(1)}_{p,q}$, we get $\pi_{\Cen}(S^{-1})\in\cl^{(0)}_{p,q}\cup\cl^{(1)}_{p,q}$. The condition $M\neq 0$ is equivalent to the condition $\pi_{\Cen}(S)\neq 0$ for the element $S\in\Pin(p,q)$ because of (\ref{t4}) and these facts.

We have the following well-known fact (see, for example \cite{Lounesto}) about the center of the Clifford algebra $\cl_{p,q}$:
\begin{eqnarray}
\Cen(\cl_{p,q})=\left\lbrace
\begin{array}{ll}
\cl^0_{p,q}=\{ue: u\in\R\}, & \mbox{if $n$ is even;}\\
\cl^0_{p,q}\oplus\cl^n_{p,q}=\{ue+u_{1\ldots n}e_{1\ldots n}: u, u_{1\ldots n}\in\R\}, & \mbox{if $n$ is odd.}
\end{array}\nonumber
\right.\nonumber
\end{eqnarray}
We have
$$(e_{1 \ldots n})^2=(-1)^{\frac{n(n-1)}{2}+q}e=\left\lbrace
\begin{array}{ll}
e, & \mbox{if $p-q=1\mod 4$;}\\
-e, & \mbox{if $p-q=3\mod 4$}
\end{array}\nonumber
\right.$$
and
$$\Cen(\cl_{p,q})\simeq\left\lbrace
\begin{array}{ll}
\R, & \parbox{.5\linewidth}{if $n$ is even;}\\
\R\oplus\R, & \parbox{.5\linewidth}{if $p-q=1\mod 4$;}\\
\C, & \parbox{.5\linewidth}{if $p-q=3 \mod 4$.}.
\end{array}\nonumber
\right.$$
Note that $\R$ and $\C$ are fields, but the set of double numbers (split-complex numbers) $\R\oplus\R$ is not a field. It is an associative commutative algebra of dimension 2 over real numbers and it has zero divisors.

Now let us consider the cases of even and odd $n$ separately.

In the case of even $n$, we have $\pi_{\Cen}(S^{-1})=\pi_0(S^{-1})$. Let us take reversion of both sides of (\ref{t4}). We obtain
\begin{eqnarray}
2^n \pi_{0}(S^{-1}) \widetilde{S}=\widetilde{M}.\label{t5}
\end{eqnarray}
Multiplying both sides of (\ref{t5}) by both sides of (\ref{t4}) we obtain
\begin{eqnarray}
2^n 2^n \pi_0(S^{-1}) (\widetilde{S}S) \pi_0(S^{-1})=\widetilde{M} M.\label{t6}
\end{eqnarray}
We have $\alpha:=\widetilde{S}S=\pm e\in\cl^0_{p,q}$ and get
\begin{eqnarray}
(2^n \pi_0(S^{-1}) )^2=\alpha\widetilde{M} M.\label{t7}
\end{eqnarray}
Taking square root and substituting this expression into (\ref{t4}), we get (\ref{1}).

In the case of odd $n$, we have $\pi_{\Cen}(S^{-1})=\pi_0(S^{-1})+\pi_n(S^{-1})$ and
\begin{eqnarray}
\widetilde{e_{1\ldots n}}=(-1)^{\frac{n(n-1)}{2}}e_{1\ldots n}=\left\lbrace
\begin{array}{ll}
e_{1\ldots n}, & \mbox{if $n=1\mod 4$;}\\
-e_{1\ldots n}, & \mbox{if $n=3\mod 4$,}
\end{array}\nonumber
\right.\nonumber\\
\widehat{\widetilde{e_{1\ldots n}}}=(-1)^{\frac{n(n-1)}{2}+n}e_{1\ldots n}=\left\lbrace
\begin{array}{ll}
e_{1\ldots n}, & \mbox{if $n=3\mod 4$;}\\
-e_{1\ldots n}, & \mbox{if $n=1\mod 4$.}
\end{array}\nonumber
\right.\nonumber
\end{eqnarray}

In the case $n=1\mod 4$, taking reversion of both sides of (\ref{t4}), we get
\begin{eqnarray}
2^n (\pi_{0}(S^{-1})+\pi_n(S^{-1})) \widetilde{S}=\widetilde{M}.\label{t8}
\end{eqnarray}
Multiplying both sides of (\ref{t8}) by both sides of (\ref{t4}), we obtain
\begin{eqnarray}
2^n 2^n  (\pi_{0}(S^{-1})+\pi_n(S^{-1}))  (\widetilde{S}S) (\pi_{0}(S^{-1})+\pi_n(S^{-1}))=\widetilde{M} M.\label{t9}
\end{eqnarray}
We have $\alpha:=\widetilde{S} S=\pm e\in\cl^0_{p,q}$ and get
\begin{eqnarray}
(2^n  (\pi_{0}(S^{-1})+\pi_n(S^{-1}))  )^2=\alpha\widetilde{M} M.\label{t10}
\end{eqnarray}
Taking square root and substituting this expression into (\ref{t4}), we get (\ref{2}) for the first case (\ref{cas}).

In the case $n=3\mod 4$, taking superposition of reversion and grade involution (it is called Clifford conjugation) of both sides of (\ref{t4}), we get
\begin{eqnarray}
2^n (\pi_{0}(S^{-1})+\pi_n(S^{-1})) \widehat{\widetilde{S}}=\widehat{\widetilde{M}}.\label{t11}
\end{eqnarray}
Multiplying both sides of (\ref{t11}) by both sides of (\ref{t4}), we obtain
\begin{eqnarray}
2^n 2^n  (\pi_{0}(S^{-1})+\pi_n(S^{-1}))  ( \widehat{\widetilde{S}}S) (\pi_{0}(S^{-1})+\pi_n(S^{-1}))=\widehat{\widetilde{M}} M.\label{t12}
\end{eqnarray}
We have $\alpha:=\widehat{\widetilde{S}}S=\pm e\in\cl^0_{p,q}$ and get
\begin{eqnarray}
(2^n  (\pi_{0}(S^{-1})+\pi_n(S^{-1}))  )^2=\alpha\widehat{\widetilde{M}} M.\label{t13}
\end{eqnarray}
Taking square root and substituting this expression into (\ref{t4}), we get (\ref{2}) for the second case (\ref{cas}).

The theorems are proved.
\end{proof}

\section{Calculation of rotors in geometric algebra}
\label{sec6}

Let us consider the particular case of Theorems \ref{th1} and \ref{th2} for the elements of the group $\Spin_+(p,q)$ and the corresponding group $\SO_+(p,q)$. Elements of $\Spin_+(p,q)$ are often called rotors and have wide application in geometric algebra (see, for example, \cite{Hestenes}, \cite{Hestenes2}, \cite{Lasenby}, \cite{bayro2}, \cite{dorst}, \cite{DLL}, \cite{LL}). We obtain the following statement.

\begin{cor}\label{th3} Let we have
$$S e_a \widetilde{S}=\beta_a,\qquad \widetilde{S}=S^{-1},$$
where two frames $e_a$ and $\beta_a$, $a=1, \ldots, n$, are related by a rotation.

If
\begin{eqnarray}
M=\beta_A e^A=e+\beta_a e^a +\cdots+\beta_{1\ldots n}e^{1\ldots n}\neq 0,\label{cond3}
\end{eqnarray}
then
$$S=\pm \frac{M}{\sqrt{\widetilde{M}M}}.$$
\end{cor}
Using our previous notation, we can write
$$\beta_a=p_a^b e_b,\quad M=\beta_A e^A=p_A^B e_B e^A,$$
$$P=||p_a^b||\in\SO_+(p,q),\quad S\in\Spin_+(p,q).$$
In the particular case of $n=3$, we have
\begin{eqnarray}
M=e+\beta_a e^a+\beta_{a_1 a_2}e^{a_1 a_2}+\beta_{123}e^{123}=2(e+ \beta_a e^a)\label{rot}
\end{eqnarray}
because $\beta_{123}=e_{123}\in\Cen(\cl_{p,q})$. We can remove scalar ``$2$'' in (\ref{rot}) because of normalization and finally obtain the following well-known formula for the rotor $S$ (see, for example,  p. 103 in \cite{Lasenby} or p. 72 in \cite{bayro2})
\begin{eqnarray}
M=e+\beta_a e^a,\qquad S=\pm \frac{M}{\sqrt{\widetilde{M}M}}.\label{las}
\end{eqnarray}
This formula is widely used in different applications of geometric algebra. Corollary \ref{th3} generalizes (\ref{las}) to the case of arbitrary $n$.

\bigskip

The results of this paper were reported at the conference AGACSE 2018 (Campinas, Brazil, July 2018). We hope that these results will be useful for different applications in computer science, robotics, and engineering. There are well-known methods of calculating of rotors in dimensions $n=3$ and $4$, but we often need geometric algebra of higher dimensions. For example, the conformal geometric algebra $\cl_{4,1}$ of dimension $n=5$ is widely used in different applications (see, for example, \cite{bayro}, \cite{dorst}, \cite{hildenbrand}, \cite{HitzerConf}). The results of this paper allow us to calculate rotors in arbitrary dimension $n=p+q\geq 1$.

% ------------------------------------------------------------------------

\subsection*{Acknowledgment}

The author is grateful to N. Marchuk and participants of the AGACSE 2018 Conference for fruitful discussions.
The author is grateful to two anonymous reviewers for their careful reading of the paper and helpful comments.
This work is supported by the Russian Science Foundation (project 18-71-00010).

% ------------------------------------------------------------------------
\end{document}